\documentclass[aps,pra,twocolumn]{revtex4-2}

\usepackage{amsmath,amsfonts}
\usepackage{amssymb,mathtools}
\usepackage{amsthm}
\usepackage{mathrsfs}
\usepackage{bbm,bm}

\newtheorem{theorem}{Theorem}
\newtheorem{definition}[theorem]{Definition}
\newtheorem{lemma}[theorem]{Lemma}
\newtheorem{proposition}[theorem]{Proposition}
\newtheorem{corollary}[theorem]{Corollary}
\newcommand{\sfTr}[1]{\mathrm{Tr}\left\{#1\right\}}
\newcommand{\sfTrAbs}[1]{\mathrm{TrAbs}\left\{#1\right\}}
\newcommand{\opTr}[1]{\mathrm{Tr}\left\{#1\right\}}
\newcommand{\opTrAbs}[1]{\mathrm{TrAbs}\left\{#1\right\}}
\newcommand{\bbTr}[1]{\mathrm{Tr}\left\{#1\right\}}
\newcommand{\bbTrAbs}[1]{\mathrm{TrAbs}\left\{#1\right\}}

\def\sfA{\mathsf{A}}
\def\sfB{\mathsf{B}}
\def\sfc{\mathsf{c}}
\def\sfK{\mathsf{K}}

\def\sfR{\mathsf{R}}
\def\sfV{\mathsf{V}}
\def\sfH{\mathsf{H}}

\def\sfZ{\mathsf{Z}}
\def\sfw{\mathsf{w}}
\def\sfW{\mathsf{W}}
\def\bbH{\mathbb{H}}
\def\bbS{\mathbb{S}}
\def\bbL{\mathbb{L}}
\def\bbU{\mathbb{U}}
\def\bbV{\mathbb{V}}
\def\bbX{\mathbb{X}}

\def \be {\begin{equation}}
\def \ee {\end{equation}}

\newcommand{\I}{\mathrm{i}\,}
\newcommand{\tr}[1]{\mathrm{tr}\left(#1\right)}
\newcommand{\Tr}[1]{\mathrm{Tr}\left\{#1\right\}}
\newcommand{\Trabs}[1]{\mathrm{TrAbs}\left\{#1\right\}}
\newcommand{\Det}[1]{\mathrm{Det}\left\{#1\right\}}

\def \Re{\mathrm{Re}\,}
\def \Im{\mathrm{Im}\,}

\def \e{\mathrm{e}}
\def \m{\mathrm{m}}
\def \e1{(\mathrm{e})}
\def \m1{(\mathrm{m})}

\def \del{\partial}

\def \Lra{\Leftrightarrow}

\def \cL{{\cal L}}

\def \cH{{\cal H}}

\def \cX{{\cal X}}
\def \sofh{{\cal S}({\cal H})}

\def \bbr{{\mathbb R}}
\def \bbc{{\mathbb C}}

\def \sofc2{{\cal S}({\mathbb C}^2)}

\def \prior{\pi(\theta)}

\def \lprior{\ell_{\pi}}
\def \htheta{\hat{\theta}}
\def \cvT{{\cal C}_\mathrm{vT}}
\def \cNH{{\cal C}_\mathrm{NH}}
\def \cN{{\cal C}_\mathrm{N}}
\def \cHo{{\cal C}_\mathrm{H}}

\def \cSLD{{\cal C}_\mathrm{SLD}}
\def \cRLD{{\cal C}_\mathrm{RLD}}

\def \rhoB{S_{\rm B}}
\def \MBj{{D}_{\mathrm{B},j}}
\def \MBk{{D}_{\mathrm{B},k}}
\def \Mtheta{D(\theta)}
\def \Mthetaj{D_{j}(\theta)}
\def \MthetaT{D^\intercal(\theta)}
\def \mutwo{\mathsf{m}}
\def \muTwo{\mathsf{M}}
\def \sfZB{\mathsf{Z}_{\mathrm{B}}}
\def \sfHB{\mathsf{H}_{\mathrm{B}}}

\def \bbSbar{\overline{\bbS}}
\def \Mbar{\overline{D}}
\def \Mbarj{\overline{D}_j}
\def \sfwbar{\overline{\sfw}}
\def \sfwtheta{\sfw(\theta)}

\begin{document}
\title{Bayesian Nagaoka-Hayashi Bound for Multiparameter Quantum-State Estimation Problem}
\author{Jun Suzuki}
\affiliation{Graduate School of Informatics and Engineering,
The University of Electro-Communications,
1-5-1 Chofugaoka, Chofu, Tokyo 182-8585, Japan.
Email: junsuzuki@uec.ac.jp}
\date{\today}

\begin{abstract}
In this work we propose a Bayesian version of the Nagaoka-Hayashi bound when estimating a parametric family of quantum states. This lower bound is a generalization of a recently proposed bound for point estimation to Bayesian estimation. We then show that the proposed lower bound can be efficiently computed as a semidefinite programming problem. As a lower bound, we also derive a Bayesian version of the Holevo-type bound from the Bayesian Nagaoka-Hayashi bound. Lastly, we prove that the new lower bound is tighter than the Bayesian quantum Cram\'{e}r-Rao bounds.
\end{abstract}
%\begin{keyword}
%Quantum-state estimation, Holevo bound, Nagaoka-Hayashi bound, Bayesian risk, semidefinite programming
%\end{keyword}
\maketitle

\section{Introduction}
Bayesian parameter estimation with prior knowledge on unknown parameters naturally enters when estimating signals upon communication processes \cite{vTbook1}.
Quantum communication is a promising near term communication technology, 
which can transmit information more securely and efficiently than the classical protocols. 
There have been many investigation on how classical and/or quantum information can be transmitted faithfully over a given noisy quantum channel, see for example \cite{gt07,wilde13,holevo19}.
Quantum Bayesian estimation is a key ingredient for decoding classical information encoded in quantum states efficiently. Quantum Bayesian estimation also has gotten a great attention in the field of quantum sensing and quantum metrology \cite{jd15,dbgmb19,gsp21,nsp21}. 

Quantum Bayesian estimation was initiated about half a century ago by Personick \cite{personick_thesis,personick71}. Due to recent advances in the quantum estimation theory, the quantum Bayesian estimation problem has gotten a renew interest by the community.
Several quantum Bayesian bounds were proposed for the Bayes risk, see for example,  \cite{personick_thesis,personick71,hlg70,hayashi_review,tsang_zk,lt16,rd2020,demkowicz2020,tsang_gl}.
However, most of them do not capture the genuine quantum nature,
since known lower bounds are based on almost a direct translation of classical Bayesian bounds. 
In particular, previously proposed lower bounds are derived by applying the Cauchy-Schwarz type inequality with respect to a certain choice for inner products on an operator space. 
Holevo initiated the investigation of nontrivial lower bounds for quantum estimation in the context of the general statistical decision problems \cite{holevo_1973}. 
He also analyzed lower bounds for the Bayes risk based on the quantum Fisher information matrices \cite{holevo_bayes,holevo_monograph,holevo_qest}. 
In particular, he gave a thorough analysis on the Gaussian shift model in the Bayesian setting.  

When estimating non-random parameters, the Holevo bound established the unique nature of the quantum estimation theory \cite{holevo,nagaoka89}. 
This is because it is expressed as a certain optimization problem without use of any quantum Fisher information matrix. 
Later, Nagaoka proposed a tighter bound for two-parameter estimation \cite{nagaoka89}. 
This lower bound is based on a different statistical problem in which one aims at finding an approximated diagonalization of two noncommuting matrices \cite{nagaoka91}. 
This Nagaoka's result was generalized by Hayashi for any finite number of noncommuting matrices \cite{hayashi99}. 
In a recent paper \cite{lja2021}, the Nagaoka bound for parameter estimation was generalized to estimating any number of non-random parameters, which was named as the Nagaoka-Hayashi bound. 
In this paper, we attempt to make a further step toward developing genuine quantum bounds based on the Nagaoka-Hayashi bound in Bayesian parameter estimation. 
In particular, we propose two Bayesian versions of these lower bounds, the Nagaoka-Hayashi bound and the Holevo-type bound. 
The unique nature of the proposed lower bounds is that they are expressed as a certain optimization problem.
%We show that it is tighter than the recently proposed a lower bound by Rubio and Dunningham \cite{rd2020}.

This paper is organized as follows. Section 2 gives a brief summary of existing quantum Bayesian bounds. 
In Sec.~3, we propose two new quantum Bayesian bounds, and we also show that the proposed lower bound is tighter than the 
previously known Bayesian quantum Cram\'{e}r-Rao bounds.
%generalized Personick bound. 
In Sec.~4, we conclude and list some open problems. 
Technical lemmas are given in Appendix. 

\section{Preliminaries}
Let $\cH$ be a finite dimensional Hilbert space and denote by $\sofh$ the totality of density matrices on $\cH$.
A quantum parametric model is a smooth family of density matrices on $\cH$, $\{S_\theta\,|\,\theta\in\Theta\}$, which is parametrized by $n$-parameter $\theta=(\theta_1,\theta_2,\ldots,\theta_n)$.
In the following, we consider a regular model, in particular, $S_\theta$ is full rank for all $\theta\in\Theta$ and 
the map $\theta\mapsto S_\theta$ is smooth and one-to-one. 
% the state is differentiable with respect to the parameter sufficiently many times. 
A measurement is described by a set of positive semidefinite matrices $\Pi_x$ where the index $x$ corresponds to a measurement outcome.
The set of operators corresponding to a quantum measurement is normally called a positive operator-valued measure (POVM), and it is defined by
\begin{equation*}
  \Pi=\{\Pi_x\};\ \forall x\in\cX,\Pi_x\ge0,\ \sum_{x\in\cX}\Pi_x=I,
\end{equation*}
where the measurement outcomes are labeled by an arbitrary set $\cX$ and $I$ is the identity operator on $\cH$. 
When the measurement outcomes are labelled with a continuous set, 
the condition on the POVM elements is $\forall x,\Pi_x\ge0,\ \int_\cX dx\,\Pi_x=I$. 

Measurement outcome is described by a random variable $X$ that obeys the conditional probability distribution:
$p_\theta(x)=\opTr{S_\theta \Pi_x}$ where $\opTr{\cdot}$ denotes the trace on $\cH$. 
The expectation value for a random variable $X$ is denoted by $E_\theta[X|\Pi]=\sum_x xp_\theta(x)$.
To infer the parameter value, we use an estimator that returns values on the set $\Theta$:
$\htheta=(\htheta_i):\ \cX\to\Theta$.
The performance of the estimator is quantified by a loss function:
\begin{equation*}
  L:\ \Theta\times \Theta\to\{x\in\bbr|x\ge0\}\cup\{\infty\}.
\end{equation*}
In this study, we adopt $L(\theta,\htheta)=\sum_{i,j}(\htheta_i-\theta_i)\sfW_{ij}(\theta)(\htheta_j-\theta_j)$.
Here $\sfW(\theta)=[\sfW_{ij}(\theta)]$ is an $n\times n$ positive semidefinite matrix, called a weight (cost) matrix.
As a special case, $W$ can be parameter independent. 
In the language of statistical decision theory, the set $(\Pi,\htheta)$ is called a quantum decision. 

The main objective of parameter estimation about quantum states $\{S_\theta\}$ is to
find the best quantum decision $(\Pi,\htheta)$ that minimizes the loss function.
As the measurement outcomes are random, we need to further identify a risk for this optimization.
\begin{definition}
The Bayes risk for a given prior probability distribution $\prior$ on $\Theta$ is defined by
\begin{equation}\label{def:Brisk1}
\sfR[\Pi,\htheta]:=\int_\Theta d\theta\, \prior E_\theta\big[L\big(\theta,\htheta(X)\big)\big|\Pi\big].
%=\sfTr{\sfW \sfV}p_{\mathrm prior}(\theta)\,d\theta,
\end{equation}
\end{definition}
With this quantum Bayes risk, the objective is to find the best quantum decision that minimizes the risk, i.e. the minimization problem over all possible quantum decisions $(\Pi,\htheta)$: 
\begin{equation}
\inf_{\Pi,\htheta}\sfR[\Pi,\htheta].
\end{equation}
%When no confusion arises, $\Pi,\htheta$ will omitted for notational simplicity.
Denoting the joint distribution by $p(\theta,x):=\prior p_\theta(x)=\tr{\prior S_\theta\Pi_x}$,
the Bayes risk \eqref{def:Brisk1} is also written as
\begin{equation*}
  \sfR[\Pi,\htheta]=E_p\big[L\big(\theta,\htheta(X)\big)\big],
\end{equation*}
where $E_p[\cdot]$ denotes the expectation value with respect to the joint distribution $p(\theta,x)$.

In the following discussion, we will also express the Bayes risk in terms of the mean square error (MSE) matrix whose $(j,k)$ component is defined by
\begin{equation*}
  \sfV_{\theta,jk}[\Pi,\htheta]:=E_\theta \big[(\htheta_j(X)-\theta_j)(\htheta_k(X)-\theta_k)\big|\Pi\big].
\end{equation*}
This then gives an alternative expression for the Bayes risk:
\begin{equation}\label{def:Brisk2}
  \sfR [\Pi,\htheta]= \int_\Theta d\theta\,\prior\sfTr{\sfW(\theta)\sfV_\theta[\Pi,\htheta]},
\end{equation}
where $\sfTr{\cdot}$ denotes the trace for matrices on the $n$-dimensional parameter space.

\subsection{Quantum van Tree inequality}
The classical van Tree inequality is based on the covariance inequality \cite{vTbook1}.
This inequality is applicable to bound the Bayesian MSE matrix 
\begin{equation*}
  \sfV_{\mathrm{B}}[\Pi,\htheta]:=\int_\Theta d\theta\,\prior\sfV_\theta[\Pi,\htheta].
\end{equation*}
Assume mild regularity conditions on the prior and the parametric model \cite{gill_levit}; 
$\pi$ and $p_\theta$ are absolutely continuous, and $\pi$ vanishes at the boundaries of the parameter set $\Theta$. 
The resulting matrix inequality is
\begin{align*}
  \sfV_{\mathrm{B}}[\Pi,\htheta]&\ge \left(J_{\mathrm{B}}[\Pi]\right)^{-1}, \\
  J_{\mathrm{B}}[\Pi]&:=J(\pi)+\int_\Theta d\theta\,\prior J_\theta[\Pi],\\
  J_{ij}(\pi)&:=\int_\Theta d\theta\,\prior \frac{\del \lprior}{\del \theta_i}\frac{\del \lprior}{\del \theta_j},
\end{align*}
where $\lprior(\theta):=\log \prior$ and $J_\theta[\Pi]$ is the Fisher information matrix about the distribution $p_\theta(x)$.
The van Tree inequality can be generalized in order to include the parameter dependent weight matrix.
This can be accomplished by the use of the Gill-Levit bound \cite{gill_levit}.
A quantum version of the Gill-Levit bound was recently proposed \cite{tsang_gl}.

It was Personick who proposed a quantum version of the van Tree inequality for the Bayes risk \cite{personick_thesis,personick71}.
Without going into the details, we give his result. 
It is known that the Fisher information matrix is bounded by the quantum Fisher information matrix $J^{\mathrm{Q}}_\theta$.
Using this fact, one can derive the quantum van Tree inequality:
\begin{align}\label{eq:QvTineq}
  \sfV_{\mathrm{B}}[\Pi,\htheta]&\ge (J^{\mathrm{Q}}_{\mathrm{B}})^{-1},\\ 
  J^{\mathrm{Q}}_{\mathrm{B}}&:=J(\pi)+\int_\Theta d\theta\,\prior J^{\mathrm{Q}}_\theta.\nonumber
\end{align}
With this inequality, one gets a lower bound for the Bayes risk \eqref{def:Brisk2} when
the weight matrix $\sfW$ is parameter independent.
\begin{equation*}
  \sfR[\Pi,\htheta]\ge \cvT:=\sfTr{\sfW (J^{\mathrm{Q}}_{\mathrm{B}})^{-1}}.
\end{equation*}
Well-known examples for the quantum Fisher information matrices are the symmetric logarithmic derivative (SLD) Fisher information matrix \cite{helstrom67} and the right logarithmic derivative (RLD) Fisher information matrix \cite{yuen-lax}. 
Originally, this lower bound \eqref{eq:QvTineq} was proven in Personick's thesis \cite[Sec.2.2.2]{personick_thesis} immediately after the pioneering work by Helstrom that formulated point estimation about quantum states \cite{helstrom67}.
In his thesis, he also derived quantum versions of Bhattacharyya and Barankin bounds, however his results seem too early to be appreciated by the community at the time.

\subsection{Bayesian quantum Cram\'{e}r-Rao bounds}
\subsubsection{Bayesian SLD Cram\'{e}r-Rao bound}
Personick also proposed a different method to derive a lower bound for the Bayes risk in the same paper where he proposed the quantum van Tree inequality \cite{personick71}.
In the published paper, he considered one-parameter estimation, and then he proved this lower bound is tight. However, it is less known that he also derived a lower bound for the general $n$-parameter estimation problem \cite{personick_thesis}.

Define the averaged states and the first moment by
\begin{align}
\rhoB&:=\int_\Theta d\theta\,\prior S_\theta,\\
\MBj&:=\int_\Theta  d\theta\,\prior\theta_j S_\theta. 
\end{align}
Next, consider a set of Hermitian matrices $L_j$ satisfying the so-called Bayesian version of the SLD equation:
\begin{equation}
\MBj=\frac12(\rhoB L_j+L_j\rhoB ).
\end{equation}
For a positive definite averaged state $\rhoB$, the Bayesian SLD $L_j$ is uniquely defined by the solution.
Then, the real symmetric matrix $\sfK=[\sfK_{jk}]$, the Bayesian SLD Fisher information matrix, is defined by
\begin{equation}
{\sfK}_{jk}=\langle L_j,L_k\rangle_{\rhoB}
=\frac12\opTr{\rhoB(L_jL_k+L_kL_j)}.
\end{equation}
Here, $\langle X,Y\rangle_{\rhoB}:=\opTr{\rhoB(X^\dagger Y+YX^\dagger)}/2$
is the symmetrized inner product for linear operators $X,Y$ on $\cH$ with respect to the state $\rhoB$, 
and $X^\dagger$ denotes the Hermitian conjugation of $X$.

For one parameter estimation, Personick proved the following inequality \cite{personick71}.
\begin{equation*}
  \sfV_{\mathrm{B}}[\Pi,\htheta]\ge \mutwo-\sfK,
\end{equation*}
where $\mutwo=\int d\theta \,\prior \theta^2$ is the second moment. 
Since the random parameter $\theta$ is scalar, 
the second term is $\sfK=\opTr{\rhoB L_{1}^{2}}$. 
%(The original form of the second term is written as $\opTr{\blue{\rhoB M_{\mathrm{B},1}}}$ in Personick's work \cite{personick71}.)

Almost half century after the seminal work by Personick,
Rubio and Dunningham generalized the Personick bound
based on a different approach \cite{rd2020}.
They proved
\begin{equation}\label{def:RDbound}
  \sfV_{\mathrm{B}}[\Pi,\htheta] \ge \muTwo-\sfK,
\end{equation}
where the first term $\muTwo$ is an $n\times n$ real symmetric matrix whose $(j,k)$ component is defined by
\begin{equation*}
\muTwo_{jk}=\int_\Theta d\theta\,\prior\theta_j\theta_k.
\end{equation*}
This is the second-moment matrix of the random variable $\theta$.
The second term is a contribution regarded as the quantum nature, and can be interpreted as the Bayesian version of the SLD Fisher information matrix.
Their bound, Eq.~\eqref{def:RDbound} with a parameter-independent weight matrix, then takes the form
\begin{equation*}
\cSLD=\sfTr{\sfW (\muTwo-{\sfK})}.
\end{equation*}
In the following, we call this lower bound as the Bayesian SLD Cram\'{e}r-Rao bound.
%generalized Personick bound.

\subsubsection{Bayesian RLD Cram\'{e}r-Rao bound}
The above construction of the Bayesian SLD Fisher information invites us to utilize the Bayesian version of other  quantum Fisher information matrices. 
Holevo proposed a bound based on the Bayesian RLD Fisher information matrix for the parameter-independent weight matrix case \cite{holevo_1973}. 
He applied it to the general quantum Gaussian system and derived the optimal estimator when the prior distribution is also a Gaussian distribution \cite{holevo_monograph}. 
Below we state the result only and its derivation can be found in Ref.~\cite{holevo_1973}. 

Let $\rhoB$ and $\MBj$ the Bayesian averaged state and the first moment defined before. 
We define the Bayesian RLD $\widetilde{L}_{j}$ by the RLD-type equation:
\begin{equation}\label{def:BRLD}
\MBj=\rhoB \widetilde{L}_{j}.
\end{equation}
Under the assumption of positivity of states, we have the unique solution $\widetilde{L}_{j}=\rhoB^{-1}\MBj$. 
Define an Hermitian $n\times n$ matrix $\widetilde{K}$ whose $(j,k)$ component is 
\begin{equation}
\widetilde{\sfK}_{jk}:=\opTr{\rhoB \widetilde{L}_{k}\widetilde{L}_{j}^\dagger}. 
\end{equation}
The Bayesian version of the RLD Cram\'{e}r-Rao bound is given as
\begin{equation}\label{def:BRLDCR}
\cRLD:=\sfTr{\sfW (\muTwo-\Re\widetilde{\sfK})}+\sfTrAbs{\sfW \,\Im\widetilde{\sfK}},
\end{equation}
where $\Re A$ ($\Im A$) denotes the component-wise real (imaginary) part of $A$ for a matrix $A\in\bbc^{n\times n}$, 
and $\sfTrAbs{A}$ denotes the sum of the absolute values for the eigenvalues of $A$. 

\section{Results}
In this section, we propose a new lower bound for the Bayes risk, which is applicable for arbitrary weight matrix. To state the main theorem, we need to introduce several notations first. 
\subsection{Alternative expression for the Bayes risk}
Noting the MSE matrix $\sfV_\theta$ consists of four terms
\begin{multline*}
\sfV_{\theta,jk}[\Pi,\htheta]=\sum_x \htheta_{j}(x)\htheta_{k}(x)\opTr{S_\theta\Pi_x}\\
\hspace{2cm}-\sum_x \theta_j\hat{\theta}_{k}(x)\opTr{S_\theta\Pi_x}\\
-\sum_x \hat{\theta}_{j}(x)\opTr{S_\theta\Pi_x}\theta_k
+\theta_j\theta_k,
\end{multline*}
we introduce the following Hermitian matrices on $\cH$:
\begin{align*}
\mathbb{L}_{jk}[\Pi,\hat{\theta}]&=\sum_x \hat{\theta}_{j}(x)\Pi_x\hat{\theta}_{k}(x)\quad(j,k=1,2,\dots,n),\\
X_j[\Pi,\hat{\theta}]&=\sum_{x}\hat{\theta}_{j}(x)\Pi_x \quad(j=1,2,\ldots,n).
\end{align*}
Importantly, these matrices are solely defined by a quantum decision $\Pi,\hat{\theta}$, and hence 
they are independent of model parameter $\theta$. 
In the following, we often omit the argument $\Pi$ and $\hat{\theta}$, when it is clear from the context. 
With these definitions, the MSE matrix is expressed as
\begin{equation*}
\sfV_{\theta,jk}[\Pi,\htheta]=\opTr{S_\theta(\bbL_{jk}-\theta_jX_k-X_j\theta_k+\theta_j\theta_k)}.
\end{equation*}

We next define a matrix and a vector which are defined on the extended Hilbert space \cite{lja2021}.
Consider the Hilbert space $\mathbb{H}:=\bbc^n\otimes\cH$, and define $ \mathbb{L}$ on $\mathbb{H}$ whose $(j,k)$ block -matrix component is given by $\mathbb{L}_{jk}$. 
We also define a column vector $X$ whose $j$th component is $X_j$ by $X=(X_1,X_2,\ldots,X_n)^\intercal$.
Here, we denote transpose of matrices and vectors with respect to $\bbc^n$ by $(\cdot)^\intercal$.
The fundamental inequality is stated in the following lemma, which is a variant of Holevo's lemma \cite{holevo}. 
\begin{lemma}[\cite{hayashi99}] \label{lem:holevo}
For all POVMs and estimators, $\mathbb{L}$ and $X$ obey the matrix inequality:
\begin{equation*}
\mathbb{L}[\Pi,\hat{\theta}]\ge X[\Pi,\hat{\theta}] \left(X[\Pi,\hat{\theta}]\right)^\intercal .
\end{equation*}
\end{lemma}

The weighted trace of the MSE matrix then takes of the form:
\begin{align*}
&\sfTr{\sfW(\theta) \sfV_\theta}\\
&\quad =\sum_{j,k}\sfW_{jk}(\theta)\opTr{S_\theta(\bbL_{jk}-\theta_jX_k-X_j\theta_k+\theta_j\theta_k)}\\
&\quad =\bbTr{\bbS\bbL}-\bbTr{\Mtheta X^\intercal}-\bbTr{X \MthetaT}
+\sfwtheta .
% &=\bbTr{\bbS_{\sfW}\bbL}-\bbTr{\theta_{\sfW} X^\intercal}-\bbTr{X \theta_{\sfW}^\intercal}
% +\sfW_\theta,
\end{align*}
%\setlength{\mathindent}{7mmm}
%where $\bbTr{\cdot}$ denotes the trace on the extended Hilbert space $\mathbb{H}$.
In this expression, we define
\begin{align*}
\bbS&:=[\bbS_{jk}(\theta)]\ \mbox{with}\ 
\bbS_{jk}(\theta):=\sfW_{jk}(\theta) S_\theta ,\\
\Mtheta&:=[\Mthetaj]\ \mbox{with}\ 
\Mthetaj:=\sum_k\sfW_{jk}(\theta)\theta_k S_\theta,\\
\sfw(\theta)&:=\sum_{j,k}\theta_j\sfW_{jk}(\theta)\theta_k.
\end{align*}
By definition, $\bbS(\theta)=\sfW(\theta)\otimes S_\theta$ is an operator on the extended Hilbert space.
$\Mtheta$ is a column vector with Hermitian matrix elements.
These quantities are determined by the quantum statistical model and the weight matrix. 

After combining the above expressions and the integration with respect to the prior, 
we obtain the alternative form of the Bayes risk:
\begin{lemma}\label{lem:Briskrep}
The Bayes risk is expressed as
\begin{align}\nonumber
\sfR[\Pi,\htheta]&=\bbTr{\bbSbar\bbL}-\bbTr{\Mbar X^\intercal}-\bbTr{X \Mbar^\intercal}+ \sfwbar ,\\
\bbSbar&:=\int_\Theta \,d\theta\,\prior\bbS,\label{def:bayesS}\\
\Mbar&:=\int_\Theta \,d\theta\,\prior\Mtheta,\label{def:bayesT}\\
\sfwbar &:=\int_\Theta \,d\theta\,\prior\sfwtheta ,\label{def:const}
\end{align}
where $\overline{\cdot}$ denotes the averaged operators with respect to the prior distribution. 
%The last quantity in Eq.~\eqref{eq:Brisk5} is model independent quantity.
\end{lemma}
We emphasize that everything is exact so far.
We also remind ourselves that $\bbL$ and $X$ are functions of a POVM $\Pi$ and an estimator $\hat{\theta}$. 

\subsection{New Bayesian bounds}
To derive a lower bound for the Bayes risk $\sfR[\Pi,\htheta]$, we follow the same line of logic used in Ref.~\cite{lja2021}. 
Combining Lemma \ref{lem:holevo} and Lemma \ref{lem:Briskrep} gives the main result of the paper.
\begin{theorem}[Bayesian Nagaoka-Hayashi bound]\label{thm:BNHbound}
For any POVM $\Pi$ and estimator $\htheta$, the following lower bound holds for the Bayes risk. 
  \begin{align}
&\sfR[\Pi,\htheta]\ge \cNH \nonumber\\
  &\cNH:=\min_{\bbL,X}\left\{ \bbTr{\bbSbar\bbL}
  -\bbTr{\Mbar X^\intercal}-\bbTr{X \Mbar^\intercal}\right\}+ \sfwbar .
  \label{def:BNHbound}
  \end{align}
Here optimization is subject to the constraints:
  $\forall j,k,\bbL_{jk}=\bbL_{kj}$, $\bbL_{jk}$ is Hermitian, $X_j$ is Hermitian, and $\bbL\geq {X} X^\intercal$.
\end{theorem}	
\begin{proof}
Let $\bbL_{*}$ and $X_{*}$ be the optimal quantities calculated from an optimal POVM $\Pi_{*}$ and an optimal estimator $\htheta_{*}$. 
Then, the following chain of inequalities holds. 
\begin{align*}
\sfR[\Pi,\htheta]&\ge\sfR[\Pi_{*},\htheta_{*}]\nonumber\\
&=\bbTr{\bbSbar\bbL_{*}}-\bbTr{ \Mbar X_{*}^\intercal}-\bbTr{X_{*} \Mbar^\intercal}+ \sfwbar \\
%&\ge \min_{\bbL\ge X_{*}X_{*}^\intercal}\left\{ \bbTr{\overline{\bbS}_{\sfW}\bbL}
%  -\bbTr{\overline{\sfS}_{\sfW} X_{*}^\intercal}-\bbTr{X_{*} \overline{\sfS}_{\sfW}^\intercal}\right\}+ \sfwbar \\
  &\ge \min_{\bbL,X}\left\{ \bbTr{\bbSbar\bbL}
  -\bbTr{\Mbar X^\intercal}-\bbTr{X \Mbar^\intercal}\right\}+ \sfwbar .
\end{align*}
The first inequality follows by definition of the optimizer. 
The second line is due to Lemma \ref{lem:Briskrep}. 
To get the third line, we apply Lemma \ref{lem:holevo}.  
In the last line, optimization is subject to the constraints stated in the theorem in particular $\bbL\geq {X} X^\intercal$. 
 
\end{proof}

The main difference from the Nagaoka-Hayashi bound for the point estimation setting \cite{lja2021} is that
there is no constraint about the locally unbiasedness.
The next result shows that the proposed Bayesian Nagaoka-Hayashi bound can be computed efficiently by a semidefinite programming (SDP) problem.  
\begin{proposition}
The Bayesian Nagaoka-Hayashi bound is semidefinite programming. 
\end{proposition}
\begin{proof}
To put the Bayesian Nagaoka-Hayashi bound in the SDP form, we write the first three terms of Eq.~\eqref{def:BNHbound} as
\begin{multline*}
\bbTr{\bbSbar\bbL}-\bbTr{\Mbar X^\intercal}-\bbTr{X \Mbar^\intercal}\\
= \bbTr{\left(\begin{array}{cc}\bbSbar & -\Mbar \\
 -\Mbar^\intercal & 0\end{array}\right)
 \left(\begin{array}{cc}\bbL & X \\ X^\intercal & 1\end{array}\right) }.
\end{multline*}
Clearly, this is an SDP problem, since the constraint on the variable $\bbL\ge {X} X^\intercal$ is equivalent to a positive semidefinite condition: 
\[
\left(\begin{array}{cc}\bbL & X \\ X^\intercal & 1\end{array}\right)\ge0. \] 
Other constraints on $\bbL$ and $X$ can also be put in the trace condition for the variable (see Supplementary Note 4 of Ref.~\cite{lja2021} how to implement these conditions as semidefinite programming). 
  
\end{proof}

As the Bayesian Nagaoka-Hayashi bound involves two optimization, 
we shall propose lower bounds, which are expressed as optimization over only $X$. 
When estimating two parameters, the Bayesian Nagaoka-Hayashi bound \eqref{def:BNHbound} 
will be called as the Bayesian Nagaoka bound as in point-wise estimation. 
This will be denoted as $\cN$.  
For the Bayesian Nagaoka bound, we show the following lower bound 
that is expressed as an optimization problem with respect to the variable $X$ only. 
\begin{theorem}[Lower bound for the Bayesian Nagaoka bound]
When the number of parameters is equal to two, the Bayesian Nagaoka bound $\cN$ is bounded as follows. 
  \begin{multline}\label{def:BNbound}
  \cN\ge\min_{X=(X_{1},X_{2})^\intercal}
  \big\{ \bbTr{\mathrm{sym}_{+}\left(\sqrt{\overline{\bbS}}XX^{\intercal}\sqrt{\overline{\bbS}}\right)}\\
  \hspace{3mm}+\int_\Theta d\theta\,\pi(\theta)\sqrt{\Det{\sfW(\theta)}}\bbTrAbs{S_\theta[X_1\,,\,X_2]}\\
%  \hspace{3mm}+\blue{\int_\Theta d\theta\,\pi(\theta)}\bbTrAbs{\mathrm{sym}_{-}\left(\sqrt{\bbS}XX^{\intercal}\sqrt{\bbS}\right)}\\
%  \big\{ \bbTr{\mathrm{sym}_{+}\left(\sqrt{\overline{\bbS}}XX^{\intercal}\sqrt{\overline{\bbS}}\right)}\\
%  \hspace{8mm}+\bbTrAbs{\mathrm{sym}_{-}\left(\sqrt{\overline{\bbS}}XX^{\intercal}\sqrt{\overline{\bbS}}\right)}\\
  -\bbTr{\Mbar X^\intercal}-\bbTr{X \Mbar^\intercal}\big\}+ \sfwbar .
  \end{multline}
Here optimization is subject to: $X_1,X_2$ are Hermitian. 
\end{theorem}
In this theorem, $\mathrm{sym}_{+}(\mathbb{A}):= \frac12 (\mathbb{A}+ \mathbb{A}^{\intercal})$ denotes the symmetrized  matrix with respect to the first Hilbert space of the extended Hilbert space $\mathbb{H}=\bbc^n\otimes\cH$, i.e., 
for $\mathbb{A}=[\mathbb{A}_{jk}]\in\mathbb{H}^{dn\times dn}$, $[\mathrm{sym}_{+}(\mathbb{A})]_{jk}=(\mathbb{A}_{jk}+\mathbb{A}_{kj})/2$. 
%In this theorem, $\mathrm{sym}_{\pm}(\mathbb{A}):= \frac12 (\mathbb{A}\pm \mathbb{A}^{\intercal})$ denotes the symmetrized (anti-symmetrized) matrix with respect to the first Hilbert space of the extended Hilbert space $\mathbb{H}=\bbc^n\otimes\cH$, i.e., 
%for $\mathbb{A}=[\mathbb{A}_{jk}]\in\mathbb{H}$, $[\mathrm{sym}_{\pm}(\mathbb{A})]_{jk}=(\mathbb{A}_{jk}\pm\mathbb{A}_{kj})/2$. 
%$\bbTrAbs{\mathbb{A}}$ denotes the sum of the absolute values for the eigenvalues of $\mathbb{A}$. 
\begin{proof}
This theorem is proven by using Lemma \ref{lem:h2} in Appendix.  
\end{proof}

We next turn our attention to the Bayesian Holevo-type bound which is in general lower than
the Bayesian Nagaoka-Hayashi bound.
\begin{theorem}[Bayesian Holevo-type bound]\label{thm:BHbound}
  \begin{multline}\label{def:BHbound}
  \cNH\ge \cHo\\
  \cHo:=
 \min_{X_j:\mathrm{Hermitian}}\big\{ \sfTr{\Re \sfZ[\overline{\bbS},X]}\\
  \hspace{20mm}+\int_\Theta d\theta\,\pi(\theta)\sfTrAbs{\Im \sfZ[\bbS(\theta),X]}\\
 -\opTr{\Mbar^\intercal X}-\opTr{X^\intercal\Mbar}\big\}+ \sfwbar ,
\end{multline}
where $\sfZ[\bbS,X]$ is an $n\times n$ Hermitian matrix is defined by
\[
\sfZ[\bbS,X]:= \mathrm{Tr}_\cH\left\{\sqrt{\bbS}XX^\intercal\sqrt{\bbS}\right\}.
\]
\end{theorem}
%In the above expression, $\Re \sfA$ ($\Im \sfA$) denotes the component-wise real (imaginary) part of a matrix $\sfA\in\bbc^{n\times n}$, 
%and $\sfTrAbs{\sfA}$ denotes the sum of the absolute values for the eigenvalues of $\sfA$. 

In the Bayesian Holevo-type bound, optimization is subject to $X_j$: Hermitian for all $j$.
This is in contrast to point estimation under the locally unbiasedness condition.
\begin{proof}
This theorem is due to Lemma \ref{lem:h3} in Appendix.  
\end{proof}

\subsection{Parameter-independent weight matrix}
To make a comparison to existing lower bounds in the literature, 
we set the weight matrix to be parameter independent.
Then, quantities \eqref{def:bayesS}, \eqref{def:bayesT}, and \eqref{def:const} are expressed as
\begin{align*}
\bbSbar&=\sfW\otimes \rhoB,\\
\Mbarj&=\sum_k\sfW_{jk}\MBk,\\
\sfwbar &=\sfTr{\sfW\muTwo},
\end{align*}
where $\rhoB=\int_\Theta \,d\theta\,\prior S_\theta$
and $\MBj=\int_\Theta \,d\theta\,\prior \theta_j S_\theta$ as before. 

In this case, $\bbSbar$ exhibits a tensor product structure as $\overline{\bbS}=\sfW\otimes \rhoB$. 
We then apply Lemma \ref{lem:h_indep} to obtain the Bayesian Holevo-type bound as follows. 
\begin{corollary}
For a parameter independent weight matrix, the Bayesian Holevo-type bound is
\begin{multline}\label{def:BHbound2}
\cHo=\min_{X_j:\mathrm{Hermitian}}\big\{ \sfTr{\sfW\,\Re \sfZB[X]}+\sfTrAbs{\sfW\,\Im \sfZB[X]}\\
-\sfTr{\sfW\sfHB[X]}-\sfTr{\sfW\sfHB[X]^\intercal}\big\}+ \sfwbar ,
\end{multline}
where $\sfZB[X]$ and $\sfHB[X]$ are
$n\times n$ matrices defined by
$\sfZ_{\mathrm{B},jk}[X]:= \opTr{\rhoB X_kX_j}$
and $\sfH_{\mathrm{B},jk}[X]:=\opTr{\MBj X_k}$, respectively.
\end{corollary}
Note that the first two terms of this Bayesian Holevo-type bound takes the same form as the point-wise version. One obtains these terms by replacing the state by the averaged state $\rhoB$. 
However, the difference appears as the contribution through the matrix $\sfHB[X]$. 

When the number of parameters is two, the Bayesian Nagaoka bound can take the following explicit form 
by applying Lemma \ref{lem:2para_indep} in Appendix. 
\begin{corollary}
For a two-parameter estimation with a parameter independent weight matrix, the Bayesian Nagaoka bound is 
expressed as 
\begin{multline}\label{def:BNbound2}
\cN=\min_{X_j:\mathrm{Hermitian}}\big\{ \sfTr{\sfW\Re \sfZB[X]}\\
\hspace{2cm}+\sqrt{\Det{\sfW}}\,\opTrAbs{\rhoB(X_{1}X_{2}-X_{2}X_{1})}\\ 
-\sfTr{\sfW\sfHB[X]}-\sfTr{\sfW\sfHB[X]^\intercal}\big\}+ \sfwbar .
\end{multline}
\end{corollary}

\subsection{Relation to the Bayesian quantum Cram\'{e}r-Rao bounds}
We claim the proposed Bayesian Nagaoka-Hayashi bound is tighter than the Bayesian quantum Cram\'{e}r-Rao bounds. 
To show this we prove that Bayesian Holevo-type bound \eqref{def:BHbound2} is larger than the Bayesian quantum Cram\'{e}r-Rao bounds. 
\subsubsection{Relation to the Bayesian SLD Cram\'{e}r-Rao bound}
First, we have the following statement. 
\begin{proposition}
$\cHo\ge\cSLD$, and hence, $\cNH\ge\cSLD$.
\end{proposition}
\begin{proof}
  First, we set $\theta$-independent $\sfW$ as the rank-1 projector $\sfW=\sfc\sfc^\intercal$ with 
  $\sfc\in\bbr^n$.
  Next, we ignore the second term in the minimization \eqref{def:BHbound2} to
  obtain a lower bound for $\cHo$.
  This gives the desired result.
\begin{align*}
  \cHo-\sfc^\intercal\muTwo\sfc&\ge \min_{X} \big\{\frac12\sum_{j,k}\sfc_j\sfc_k(\opTr{\rhoB(X_jX_k+X_kX_j)}\nonumber\\
  &\hspace{15mm}- \opTr{\MBj X_k}- \opTr{X_j\MBk})\big\}\\
  % &= \min_{X} \{\sum_{jk}\sfc_j\sfc_k(\opTr{\overline{S}X_jX_k}-
  % \langle L_j,X_k\rangle_{\overline{S}}
  % -\langle X_j,L_k\rangle_{\overline{S}}
  % )\}+\sfc^\intercal\muTwo\sfc\\
  & = \min_{X} \{\langle X_\sfc,X_\sfc\rangle_{\rhoB}
  -\langle X_\sfc,L_\sfc\rangle_{\rhoB}-\langle L_\sfc,X_\sfc\rangle_{\rhoB}
  \}\\
  & = \min_{X} \{\langle X_\sfc-L_\sfc,X_\sfc-L_\sfc\rangle_{\rhoB}\}
  -\langle L_\sfc,L_\sfc\rangle_{\rhoB}\\
  & =-\langle L_\sfc,L_\sfc\rangle_{\rhoB}\\
  & =-\sfc^\intercal\sfK\sfc,
\end{align*}
  where we set $X_\sfc=\sum_j\sfc_jX_j$ and $L_\sfc=\sum_j\sfc_jL_j$. 
  Optimization over $X_j$ is solved by the choice $X_j=L_j$ (Hermitian). 
  Since this is true for any choice of $\sfc$, we prove the matrix inequality $\sfV_\mathrm{B}\ge \muTwo-\sfK$. 
  Thereby, we obtain the relation $\cHo\ge\cSLD$.  
\end{proof}

\subsubsection{Relation to the Bayesian RLD Cram\'{e}r-Rao bound}
We next show that the proposed bound is also tighter than the Bayesian RLD Cram\'{e}r-Rao bound \eqref{def:BRLDCR}. 
A key element to prove this result is Lemma used by Holevo (see for example, Ref.~\cite[Lemma 6.6.1]{holevo}). 
This establishes the identity:
\begin{multline}\label{lem:ho}
\sfTr{\sfW\sfA}+\sfTrAbs{\sfW\sfB}=\\
\min_{\sfV:\,\mathrm{real\ symmetric}}\left\{\sfTr{\sfW \sfV}\,|\,\sfV\ge\sfA+\I\sfB \right\},
\end{multline}
for any $\sfA,\sfB,\sfW\in\bbr^{n\times n}$ such that $\sfW>0$, $\sfA^\intercal=\sfA$, and $\sfB^\intercal=-\sfB$. 
With this Lemma we can prove the next proposition. 
\begin{proposition}
$\cHo\ge\cRLD$, and hence, $\cNH\ge\cRLD$.
\end{proposition}
\begin{proof}
Notice that $\Re \sfZB[X]$ and $\sfHB[X]+\sfHB[X]^\intercal$ are real symmetric and 
$\Im \sfZB[X]$ is real antisymmetric. 
We then apply the Holevo's Lemma \eqref{lem:ho} to minimization for the expression \eqref{def:BHbound2}. 
\begin{multline*}
  \cHo-\overline{\sfw}\\
  \quad = \min_{X_j:\,\mathrm{Hermitian}} \min_{\sfV:\,\mathrm{real\ symmetric}}
  \big\{ \sfTr{\sfW \sfV}\,|\,\sfV\ge\Re\sfZB[X]\\
  -\sfHB[X]+\sfHB[X]^\intercal+\I\Im\sfZB[X]\big\}.
\end{multline*}
By the definition of the Bayesian RLD \eqref{def:BRLD}, we can rewrite the matrix 
in the condition as follows.
\begin{align*}
&\Re\sfZB[X]-\sfHB[X]+\sfHB[X]^\intercal+\I\Im\sfZB[X]\\
&\quad =\sfZB[X]-\big[\opTr{\rhoB \widetilde{L}_k{X}_j}\big]-\big[\opTr{\rhoB {X}_k\widetilde{L}_j^\dagger}\big]\\
&\quad =\sfZB[X-\widetilde{L}]-\sfZB[\widetilde{L}]\\
&\quad =\sfZB[X-\widetilde{L}]-\widetilde{\sfK}.
\end{align*}
Finally, we can carry out optimization over $X_j$ by enlarging the domain to any operators to obtain the following lower bound. 
\begin{align*}
&\cHo-\overline{\sfw}\\
&\quad \ge  \min_{\sfV:\,\mathrm{real\ symmetric}}\min_{X_j}\big\{ \sfTr{\sfW \sfV}\,|\,\sfV\ge\sfZB[X-\widetilde{L}]-\widetilde{\sfK}\big\}\\
&\quad =\min_{\sfV:\,\mathrm{real\ symmetric}}\big\{ \sfTr{\sfW \sfV}\,|\,\sfV\ge-\widetilde{\sfK}\big\}\\
%&=\sfTr{\sfW\,\Re(-\widetilde{\sfK})}+\sfTrAbs{\sfW\,\Im(-\widetilde{\sfK})}\\
&\quad =-\sfTr{\sfW\,\Re\widetilde{\sfK}}+\sfTrAbs{\sfW\,\Im\widetilde{\sfK}},
\end{align*}
where minimization about $\sfV$ is solved by Holevo's Lemma (Eq.~\eqref{lem:ho}).  
\end{proof}

\section{Conclusion and outlook}
In summary we proposed a new lower bound for the Bayes risk, called the Bayesian Nagaoka-Hayashi bound. 
This bound then can be bounded by a Bayesian version of the Holevo bound.
It was shown that the new lower bound is tighter than the existing lower bound based on the Bayesian quantum Cram\'{e}r-Rao bounds \cite{personick71,holevo_1973,rd2020}.
The proposed Bayesian lower bounds are based on the idea developed in the previous publication \cite{lja2021}. 

In this paper, we only derived bounds, and hence there are many future directions. 
Firstly, we need to analyze achievability of new proposed bounds both in the asymptotic and non asymptotic theory. 
For example, Ref.~\cite{gill08} investigated the first order asymptotics of the Bayes risk. 
Secondly, relations to other Bayesian bounds need to be examined. 
In particular, a further analysis is needed when the weight matrix is parameter-dependent. 
Thirdly, an extension including random parameters in the presence of nuisance parameters will be 
important. Non-random parameter estimation with the nuisance parameter has been investigated recently \cite{syh2021,tad2020}. 
Extensions of these formulations in the Bayesian setting will be presented in the future work. 

\section*{Acknowledgmenet} %% 
The work is partly supported by JSPS KAKENHI Grant Number JP21K11749 and JP21K04919. 
The author would like to thank Mr.~L. Conlon and Dr.~S.M. Assad for collaboration at the early stage of the project. 
We also sincerely thank anonymous reviewers for careful reading, constructive comments, and valuable suggestions to improve the manuscript.

\appendix
\section{Proofs}
In this appendix, we list necessary lemmas to prove the results. 
Some of lemmas are obtained by extending known results \cite{holevo,nagaoka89,hayashi99}, 
and we keep proofs short. 

Consider a tensor product Hilbert space $\bbH=\cH_{1}\otimes\cH_{2}$, 
where $\cH_1=\bbc^n$ and $\cH_2=\cH$ as in the main text.
Let $\cL(\bbH)$, $\cL_{+}(\bbH)$, and $\cL_{h}(\bbH)$ be the set of linear operators, positive semidefinite, and Hermitian matrices on $\bbH$, respectively. 
(Strictly positive matrix case is denoted by $\cL_{++}(\bbH)$.) 
We denote by $\cL_{\mathrm{sym}_1}(\bbH)$ the set of symmetric matrices under the partial transpose $(\cdot)^{T_1}$ with respect to the subspace $\cH_{1}$. 
$\mathrm{sym}_{\pm}(\bbL):= \frac12 (\bbL\pm \bbL^{T_1})$ denotes the symmetrized (antisymmetrized) matrix of $\bbL\in\cL(\bbH)$ with respect to $\cH_{1}$. 
We also use the notation such as 
$\cL_{\mathrm{sym}_1,+}(\bbH)=\cL_{\mathrm{sym}_1}(\bbH)\cap \cL_{+}(\bbH)$ to denote 
the intersection of two sets. 
An operator $\bbL\in\cL(\bbH)$ is also expressed as a matrix-valued matrix in the form: 
$\bbL=[\bbL_{jk}]$ where each component $\bbL_{jk}$ is a matrix on $\cH_2$. 
This representation is unique if we fix the basis of $\cH_1$. 
With this representation, $[\bbL^{T_1}]_{jk}=[\bbL_{kj}]$ and 
$[\mathrm{sym}_{\pm}(\bbL)]_{jk}=(\bbL_{jk}\pm \bbL_{kj})/2$. 

Given a strictly positive matrix $\bbS\in\cL_{\mathrm{sym}_1,++}(\bbH)$, which is symmetric under the partial transpose with respect to $\cH_1$, 
and an Hermitian matrix $\bbX$ on $\bbH$, we define the following optimization problem. 
\begin{equation}\label{eq:OptProb}
f(\bbS,\bbX):=\min_{\bbL\in \cL_{\mathrm{sym}_1,h}(\bbH)}\left\{\Tr{\bbS\bbL}\,|\, \bbL\ge \bbX \right\}.
\end{equation}
%where $\cL_{+,\mathrm{sym}_1}$ is the intersection of $\cL_{+}(\cH)$ and $\cL_{\mathrm{sym}_1}(\cH)$. 
We give an alternative expression for $f(\bbS,\bbX)$ and then give lower bounds below.

First, we consider a special case for $\bbS=\sfW\otimes S>0$ with $\sfW^T=\sfW$. 
With this assumption, a direct calculation shows the identity
\begin{equation}\label{eq:lem-indep}
\mathrm{sym}_{\pm}(\sqrt{\bbS}\bbX\sqrt{\bbS})=\sqrt{\bbS}\mathrm{sym}_{\pm}(\bbX)\sqrt{\bbS}. 
\end{equation}
This rewrites the problem as follows. 
\begin{lemma}\label{lem:h1}
For $\bbS=\sfW\otimes S\in\cL_{\mathrm{sym}_1,++}(\bbH)$, 
the optimization problem \eqref{eq:OptProb} is expressed as
\begin{multline}\label{eq:OptProb2}
f(\sfW\otimes S,\bbX)= \Tr{\sqrt{\bbS}\mathrm{sym}_{+}(\bbX)\sqrt{\bbS}}\\
+ \min_{\bbV\in \cL_{\mathrm{sym}_1,h}(\bbH)}\left\{\Tr{\bbV}\,|\,\bbV\ge \sqrt{\bbS}\mathrm{sym}_{-}(\bbX)\sqrt{\bbS}\right\}.
\end{multline}
\end{lemma}
\begin{proof}
Write the constraint of optimization as $\bbL\ge \bbX$ $\Lra$ $\sqrt{\bbS} \bbL \sqrt{\bbS}\ge \sqrt{\bbS} \bbX \sqrt{\bbS}
=\mathrm{sym}_{+}(\sqrt{\bbS}\bbX\sqrt{\bbS})+\mathrm{sym}_{-}(\sqrt{\bbS}\bbX\sqrt{\bbS})$.  
We note that the first term $\mathrm{sym}_{+}(\sqrt{\bbS}\bbX\sqrt{\bbS})=\sqrt{\bbS}\mathrm{sym}_{+}(\bbX)\sqrt{\bbS}$ is in the set $\cL_{\mathrm{sym}_1,h}(\bbH)$ due to the assumptions $\bbS=\sfW\otimes S$ and $\sfW^T=\sfW$.  
We substitute $\bbL=\mathrm{sym}_{+}(\bbX)+\bbS^{-1/2}\bbV\bbS^{-1/2}$ to get 
\[
\Tr{\bbS\bbL}=\Tr{\sqrt{\bbS}\mathrm{sym}_{+}(\bbX)\sqrt{\bbS}}+\Tr{\bbV},
\]
where optimization is subject to $\bbV\ge \sqrt{\bbS}\mathrm{sym}_{-}(\bbX)\sqrt{\bbS}$. 
Because $\bbS^{-1/2}=\sfW^{-1/2}\otimes S^{-1/2}$ preserves symmetry under the partial transpose with respect to $\cH_1$, $\bbV$ is subject to the constraint $\cL_{\mathrm{sym}_1,h}(\bbH)$.  
\end{proof}

Next, we show optimization \eqref{eq:OptProb2} can be solved with the further assumption of dim$\cH_{1}=2$. 
\begin{lemma}\label{lem:2para_indep}
When the dimension of the Hilbert space $\cH_{1}$ is two (dim$\cH_{1}=2$), 
the above optimization problem \eqref{eq:OptProb} is carried out for 
$\bbS=\sfW\otimes S\in\cL_{\mathrm{sym}_1,++}(\bbH)$.
\begin{align}\label{eq:lem-h2}
\hspace{0mm}f(\sfW\otimes S,\bbX)&=f_1(\sfW\otimes S,\bbX),\\
f_1(\sfW\otimes S,\bbX)&= \Tr{\sqrt{\bbS}\mathrm{sym}_{+}(\bbX)\sqrt{\bbS}}\nonumber\\
&\hspace{1mm}+ \sqrt{\Det{\sfW}}\Trabs{S(\bbX_{12}-\bbX_{21}) }, 
\nonumber
\end{align}
where $\Trabs{X}$ denotes the sum of absolute values for the eigenvalues of the operator $X\in \cL(\cH)$. 
\end{lemma}
\begin{proof}
When dim$\cH_{1}=2$, we explicitly express the right hand side of the constraint in optimization \eqref{eq:OptProb2} 
as follows. Define $Y:=\sqrt{\Det{\sfW}}\sqrt{S}\frac12(\bbX_{12}-\bbX_{21})\sqrt{S}$, then 
\begin{align*}
&\sqrt{\bbS}\mathrm{sym}_{-}(\bbX)\sqrt{\bbS}\\
&\quad =\sqrt{\sfW}\otimes \sqrt{S}\frac12\left(\begin{array}{cc}0 & \bbX_{12}-\bbX_{21} \\ \bbX_{21}-\bbX_{12} & 0\end{array}\right)\sqrt{\sfW}\otimes \sqrt{S}\\
&\quad =\sqrt{\sfW}\left(\begin{array}{cc}0 & -1 \\ 1 & 0\end{array}\right)\sqrt{\sfW}\otimes
\sqrt{S}\frac12(\bbX_{12}-\bbX_{21})\sqrt{S}\\
%&\ =\sqrt{\sfW}\otimes I\left(\begin{array}{cc}0 & -\I \\ \I & 0\end{array}\right)\otimes\sqrt{S}\frac12(\bbX_{12}-\bbX_{21})\sqrt{S}\sqrt{\sfW}\otimes I\\
&\quad =\sqrt{\Det{\sfW}}\left(\begin{array}{cc}0 & -1 \\ 1 & 0\end{array}\right)\otimes \sqrt{S}\frac12(\bbX_{12}-\bbX_{21})\sqrt{S}\\ 
&\quad =\left(\begin{array}{cc}0 & Y \\ -Y & 0\end{array}\right). 
\end{align*}
Note that $Y^\dagger=-Y$ holds, since $\bbX$ is Hermitian. 
This implies that $\I Y$ is an Hermitian matrix on $\cH_2$, and hence $\I Y$ has real eigenvalues. 
With this constraint, it is known that optimization can be carried out \cite{hayashi99,lja2021}. 
Here we provide a simple one. Let $u$ be a unitary matrix that diagonalizes the matrix 
{\scriptsize $\left(\begin{array}{cc}0 & -\I \\ \I & 0\end{array}\right)$}, and define 
the unitary $\bbU:=u\otimes I$. 
The block-diagonal components of the matrix $\bbU (\bbV-\sqrt{\bbS}\mathrm{sym}_{-}(\bbX)\sqrt{\bbS})\bbU^\dagger$ 
are 
\begin{equation*}
\frac12(\bbV_{11}+\bbV_{22})\pm\frac{1}{2}(Y-Y^\dagger). 
\end{equation*}
Since $\bbU (\bbV-\sqrt{\bbS}\mathrm{sym}_{-}(\bbX)\sqrt{\bbS})\bbU^\dagger\ge0$, 
the sum of block-diagonal components of $\bbV$ must satisfy
\begin{equation}
\bbV_{11}+\bbV_{22}\ge \pm (Y-Y^\dagger)=\pm 2Y. 
\end{equation}
By solving optimization with this constraint, we obtain a lower bound for the original problem. 
This is carried out by applying Lemma in Holevo's book (\cite[Lemma 6.6.1]{holevo}) as 
\begin{align*}
&\min_{\bbV\in\cL_{\mathrm{sym}_1,h}(\bbH)}\left\{\Tr{\bbV}\,|\,\bbV\ge \sqrt{\bbS}\mathrm{sym}_{-}(\bbX)\sqrt{\bbS}\right\}\\
&\qquad \ge\min_{\bbV\in \cL_{\mathrm{sym}_1,h}(\bbH)}\left\{\Tr{\bbV}\,|\, \bbV_{11}+\bbV_{22}\ge \pm 2Y\right\}\\
&\qquad =\Tr{|2Y|}\\
&\qquad =\sqrt{\Det{\sfW}}\Tr{\left|\sqrt{S}(\bbX_{12}-\bbX_{21})\sqrt{S})\right| }\\
&\qquad =\sqrt{\Det{\sfW}}\Trabs{S(\bbX_{12}-\bbX_{21}) }. 
\end{align*}
The optimizer here is given by $\bbV_{11}+\bbV_{22}=2|V|$. 
In fact, we see that this lower bound is attained by the optimizer 
$\bbV_*={\scriptsize \left(\begin{array}{cc}|Y| & 0 \\ 0 & |Y|\end{array}\right)}$, 
which is in the set $\cL_{\mathrm{sym}_1,h}(\bbH)$. 
This is because ${\scriptsize \left(\begin{array}{cc}|Y| & Y \\ -Y & |Y|\end{array}\right)}\ge0$. 
Therefore, this proves the claim.  
 
\end{proof}

As we cannot execute optimization for the original problem \eqref{eq:OptProb2} in general, 
we will derive a lower bound which takes a similar form as the Holevo bound. 
We denote the partial trace of $\bbV$ with respect to the second Hilbert space $\cH_2$ by $\sfV=\mathrm{Tr}_{2}\{ \bbV\}$. 
For any $\bbV$ in the original optimization problem, $\sfV$ becomes real symmetric. 
Let $\cL_{\mathrm{sym},h}(\cH_{1})$ the set of real symmetric matrices on $\cH_1$.  
With this notation, we can get a lower bound for the problem. 
\begin{lemma}\label{lem:h_indep}
For $\bbS=\sfW\otimes S\in\cL_{\mathrm{sym}_1,++}(\bbH)$, the next lower bound holds for $f(\sfW\otimes S,\bbX)$.
\begin{align*}
f(\sfW\otimes S,\bbX)&\ge f_2(\sfW\otimes S,\bbX),\\
f_2(\bbS,\bbX)&:= \Tr{\Re \sfZ(\bbS,\bbX)}+ \Trabs{\Im \sfZ(\bbS,\bbX)},
%\Tr{\mathrm{sym}_{+}(\sqrt{\bbS}\bbX\sqrt{\bbS})}\nonumber \\
%&\hspace{-15mm} +\min_{\sfV\in \cL_{\mathrm{sym},h}(\cH_{1})}\{\Tr{\sfV}\,|\,\sfV\ge \mathrm{Tr}_{2}\{\mathrm{sym}_{-}(\sqrt{\bbS}\bbX\sqrt{\bbS})\}\}\\
%%f_2(\sfW\otimes S,\bbX)
%&=\Tr{\Re \sfZ(\bbS,\bbX)} + \Trabs{\Im \sfZ(\bbS,\bbX)},
\end{align*} 
where $\sfZ(\bbS,\bbX):= \mathrm{Tr}_{2}\{\sqrt{\bbS}\bbX\sqrt{\bbS}\}$. 
\end{lemma}
\begin{proof}
First, we have an inequality:
\begin{multline*}
f(\sfW\otimes S,\bbX)\ge\Tr{\mathrm{sym}_{+}(\sqrt{\bbS}\bbX\sqrt{\bbS})} \\
+\min_{\sfV\in \cL_{\mathrm{sym},h}(\cH_{1})}\{\Tr{\sfV}\,|\,\sfV\ge \mathrm{Tr}_{2}\{\mathrm{sym}_{-}(\sqrt{\bbS}\bbX\sqrt{\bbS})\}\}, 
\end{multline*}
since any matrix $\bbV$ under the original constraint in $f(\bbS,\bbX)$ 
satisfies the condition $\sfV=\mathrm{Tr}_{2}\{\bbV\}\ge \mathrm{Tr}_{2}\{\mathrm{sym}_{-}(\sqrt{\bbS}\bbX\sqrt{\bbS})\}$. 
%This means that optimization for $f_2(\bbS,\bbX)$ is carried out in a larger space. 
The expression of the second line of the above inequality directly follows from Lemma \ref{lem:h1}. 
%The equivalent expression for $f_2(\bbS,\bbX)$ is shown as follows. 
Second, a direct calculation shows that 
$\Tr{\sqrt{\bbS}\mathrm{sym}_{+}(\bbX)\sqrt{\bbS}}=\Tr{\Re \sfZ(\bbS,\bbX)}$. 
Last, we apply the lemma \cite[Lemma 6.6.1]{holevo} to solve the above optimization.  
This is done by noting $\sfV\ge \mathrm{Tr}_{2}\{\mathrm{sym}_{-}(\sqrt{\bbS}\bbX\sqrt{\bbS})\}$ implies 
$\sfV\ge \pm\mathrm{Tr}_{2}\{\mathrm{sym}_{-}(\sqrt{\bbS}\bbX\sqrt{\bbS})\}$ and 
$\mathrm{Tr}_{2}\{\mathrm{sym}_{-}(\sqrt{\bbS}\bbX\sqrt{\bbS})\}=\I\Im \sfZ(\bbS,\bbX)$. 
   
\end{proof}

Next, we discuss a more general case where $\bbS$ takes a convex mixture of tensor products of the form:
\begin{equation}
\bbS=\sum_{j}\pi_j\sfW_j\otimes S_j\ \mbox{ with }\ \sfW_j^T=\sfW_j>0, 
\end{equation}
where $(\pi_j)$ is a probability distribution on a finite set. 
In this case, we cannot use the relation \eqref{eq:lem-indep} anymore, 
and Lemma \ref{lem:h1} does not hold. 
Nevertheless, similar techniques are used to obtain a lower bound for the optimization problem. 
The basic logic in the following discussion is that optimizing individual functions gives 
a lower value than optimizing the sum of the functions. 
\begin{lemma}\label{lem:h0}
For $\bbS=\sum_{j}\pi_j\sfW_j\otimes S_j\in \cL_{\mathrm{sym}_1,++}(\bbH)$, 
we have a lower bound for the problem \eqref{eq:OptProb}:
\begin{align}\label{eq:OptProb3}
&f(\bbS,\bbX)\ge f_3(\bbS,\bbX),\\
&f_3(\bbS,\bbX)= \Tr{\sqrt{\bbS}\mathrm{sym}_{+}(\bbX)\sqrt{\bbS}}\nonumber\\
&\hspace{-3mm}+ \sum_j\pi_j\min_{\bbV\in \cL_{\mathrm{sym}_1,h}(\bbH)}\left\{\Tr{\bbV}\,|\,\bbV\ge \sqrt{\bbS_j}\mathrm{sym}_{-}(\bbX)\sqrt{\bbS_j}\right\}.
\nonumber
\end{align}
\end{lemma}
\begin{proof}
Denote $\bbS_j=\sfW_j\otimes S_j$, we obtain a lower bound for the optimization problem as 
\begin{align*}
f(\bbS,\bbX)&=\min_{\bbL\in \cL_{\mathrm{sym}_1,h}(\bbH)}\left\{\sum_j\pi_j\Tr{\bbS_j\bbL}\,|\, \bbL\ge \bbX \right\}\\
&\ge \sum_j\pi_jf(\sfW_j\otimes S_j,\bbX). 
%&\ge \sum_j\pi_j\min_{\bbL\in \cL_{\mathrm{sym}_1,h}(\bbH)}\left\{\Tr{\bbS_j\bbL}\,|\, \bbL\ge \bbX \right\}. 
\end{align*}
For each optimization defined by $\bbS_j$ and $\bbX$, we apply Lemma \ref{lem:h1} to get
\begin{align*}
&f(\bbS,\bbX)\ge \sum_j \pi_j\Tr{\sqrt{\bbS_j}\mathrm{sym}_{+}(\bbX)\sqrt{\bbS_j}}\\
&\ + \sum_j\pi_j\min_{\bbV\in \cL_{\mathrm{sym}_1,h}(\bbH)}\left\{\Tr{\bbV}\,|\,\bbV\ge \sqrt{\bbS_j}\mathrm{sym}_{-}(\bbX)\sqrt{\bbS_j}\right\}.
\end{align*}
Lastly, we note the first term is also expressed as
\begin{align*}
\sum_j \pi_j\Tr{\sqrt{\bbS_j}\mathrm{sym}_{+}(\bbX)\sqrt{\bbS_j}}
&= \sum_j \pi_j\Tr{\bbS_j\mathrm{sym}_{+}(\bbX)}\\
&=\Tr{\bbS\mathrm{sym}_{+}(\bbX)}\\
&=\Tr{\sqrt{\bbS}\mathrm{sym}_{+}(\bbX)\sqrt{\bbS}}, 
\end{align*}
where linearity and cyclic property of the trace is used.   
\end{proof}

When dim$\cH_{1}=2$, we obtain the next lower bound by applying Lemma \ref{lem:2para_indep} 
in Lemma \ref{lem:h0}. 
\begin{lemma}\label{lem:h2}
When the dimension of the Hilbert space $\cH_{1}$ is two 
and $\bbS=\sum_{j}\pi_j\sfW_j\otimes S_j\in \cL_{\mathrm{sym}_1,++}(\bbH)$, 
we have a lower bound: 
\begin{align}\label{eq:OptProb4}
f(\bbS,\bbX)&\ge f_4(\bbS,\bbX)\\
f_4(\bbS,\bbX)&= \Tr{\sqrt{\bbS}\mathrm{sym}_{+}(\bbX)\sqrt{\bbS}}\nonumber\\
\hspace{0mm}+ \sum_j &\pi_j\sqrt{\Det{\sfW_j}}\Trabs{S_j(\bbX_{12}-\bbX_{21}) }.
%\hspace{0mm}+ \sum_j &\pi_j\sqrt{\Det{\sfW_j}}\Trabs{\sqrt{S_j}(\bbX_{12}-\bbX_{21})\sqrt{S_j}) }.
%\Trabs{\sqrt{\bbS_j}\mathrm{sym}_{-}(\bbX)\sqrt{\bbS_j}}. 
\nonumber
\end{align}
\end{lemma}
\begin{proof}
Use Lemma \ref{lem:h0} and optimization for each $\bbS_j=\sfW_j\otimes S_j$ is done by Lemma \ref{lem:2para_indep}, 
we obtain the desired lower bound.  
%\begin{multline*}%\label{eq:2para}
%\Tr{\sqrt{\bbS}\mathrm{sym}_{+}(\bbX)\sqrt{\bbS}}\\
%+ \sum_j \sqrt{\Det{\sfW_j}}\Trabs{\sqrt{S_j}(\bbX_{12}-\bbX_{21})\sqrt{S_j}) }.
%\end{multline*}
%To obtain the desired lower bound in terms of $\bbS$, 
%we use the triangle inequality for the second term of Eq.~\eqref{eq:2para}. 
%\begin{align}
%&\sum_j \sqrt{\Det{\sfW_j}}\Trabs{\sqrt{S_j}(\bbX_{12}-\bbX_{21})\sqrt{S_j}) }\\
%&=\sum_j \Tr{\left|\sqrt{\bbS_j}\mathrm{sym}_{-}(\bbX)\sqrt{\bbS_j} \right|}\\
%&\ge\Tr{\left|\sum_j\sqrt{\bbS_j}\mathrm{sym}_{-}(\bbX)\sqrt{\bbS_j} \right|}\\
%&=
%\end{align}
  
\end{proof}

Finally, we apply Lemma \ref{lem:h_indep} in Lemma \ref{lem:h0} to get the Holevo-type bound. 
\begin{lemma}\label{lem:h3}
For $\bbS=\sum_{j}\pi_j\sfW_j\otimes S_j\in \cL_{\mathrm{sym}_1,++}(\bbH)$, 
we have a lower bound: 
\begin{align}\label{eq:OptProb5}
f(\bbS,\bbX)&\ge f_5(\bbS,\bbX)\\
f_5(\bbS,\bbX)&= \Tr{\sqrt{\bbS}\mathrm{sym}_{+}(\bbX)\sqrt{\bbS}}\nonumber\\
&\hspace{2cm}+ \sum_j \pi_j \Trabs{\Im \sfZ(\bbS_j,\bbX)}. \nonumber
\end{align}
\end{lemma}
\begin{proof}
For each $\bbS_j$ of Lemma \ref{lem:h0}, we apply Lemma \ref{lem:h_indep}.  
\end{proof}

We remark that the above results can be extended to a convex mixture with the following integral form 
under the assumption of integrability of the matrix: 
\begin{equation}
\bbS=\int d\theta\,\pi(\theta)\sfW(\theta)\otimes S_\theta,  
\end{equation}
where $\pi$ is a probability density function, $\sfW(\theta)^T=\sfW(\theta)>0$, and $S_\theta>0$ for all $\theta$.

\end{document}